\documentclass[11pt]{article}
\usepackage[a4paper,margin=30mm]{geometry}
\usepackage[utf8]{inputenc}
\usepackage{amsmath}
\usepackage{amsthm}
\usepackage{amsfonts}
\usepackage{amssymb}
\usepackage{mathtools}
\usepackage{amsbsy}
\usepackage[hidelinks]{hyperref}
\allowdisplaybreaks
\mathtoolsset{showonlyrefs=true}

\usepackage{fourier}
\usepackage{xcolor}

\newcommand{\Vir}{\mathrm{Vir}}
\newcommand{\Z}{\mathbb{Z}}
\newcommand{\C}{\mathbb{C}}

\theoremstyle{plain}
\newtheorem{thm}{Theorem}[section]
\newtheorem{cor}[thm]{Corollary}
\newtheorem{lem}[thm]{Lemma}

\newtheorem{re}[thm]{Remark}

\numberwithin{equation}{section}

\begin{document}
\title{Existence and Uniqueness of Irregular Vectors\\ of Integer and Half-Integer Ranks for the Virasoro Algebra}
\author{Hajime Nagoya \\ School of Mathematics and Physics, Kanazawa University, \\Kanazawa, Ishikawa 920-1192, Japan \\ E-mail: \href{mailto:nagoya@se.kanazawa-u.ac.jp}{nagoya@se.kanazawa-u.ac.jp}}
\date{}
\maketitle

\begin{abstract}
   Although irregular vectors for the Virasoro algebra are widely used in modern
mathematical physics, a rigorous existence and uniqueness theorem in arbitrary
rank has not been available in the literature. 
In this paper, we develop an
algebraic framework, based on Virasoro differential operators on the parameter
space, which gives such a theorem for arbitrary integer and half-integer ranks.
A key ingredient is the construction of a canonical operator \(L_*\) from the
coefficient matrix of the vector-field part of a truncated Virasoro
realization. This operator closes the recursive system by isolating the
derivative with respect to the highest irregular parameter. Using this
mechanism, we prove the existence and uniqueness of formal irregular vectors
of arbitrary integer rank. We then construct the truncated Virasoro vector
fields required in the half-integer-rank setting and prove the existence and
uniqueness of the corresponding half-integer-rank formal irregular vectors.
We also prove that, after a scalar gauge normalization, the canonical
solutions satisfy the full lower Virasoro deformation equations.
These results provide an algebraic foundation for the rigorous construction
of irregular conformal blocks built from higher-rank irregular vectors. After passing to eigenvalue coordinates, the vector-field part of the
half-integer construction is identified with the differential realizations
appearing in the literature, while the zeroth-order terms are
explained by scalar gauge freedom.
\end{abstract}

\tableofcontents
\section{Introduction}

Irregular conformal blocks and irregular vectors of the Virasoro algebra have
played an important role in the interaction between two-dimensional conformal
field theory, gauge theory, and the theory of isomonodromic deformations. They
first appeared in the study of the AGT correspondence \cite{AGT} and its degenerations \cite{Gaiotto},
where special vectors in Virasoro Verma modules were used to describe 
conformal blocks associated with asymptotically free gauge theories. Since then, irregular conformal blocks have appeared
in a wide range of problems, including strong-coupling expansions of partition functions in Argyres-Douglas theories \cite{GT}, Painlev\'e and quantum Painlev\'e equations \cite{BST25, GIL13, LNR18, Nagoya15, Nagoya18},
confluent Heun equations \cite{BILT23, INS26, LN22}, black-hole perturbation theory \cite{CunhaCavalcante20, BILT22}, and exact WKB analysis \cite{IILZ25, I20}.

There are, however, two different kinds of objects that are often both called
irregular vectors. The first type, introduced by Gaiotto \cite{Gaiotto}, has
regular singular behavior with respect to the relevant deformation parameters.
These vectors, often called Gaiotto vectors, are characterized by Whittaker-type
conditions for the Virasoro algebra and have been studied extensively in
connection with the degenerations of the AGT correspondence \cite{FJK12, HJS10}. The second type, which is the subject of
the present paper, has genuinely irregular singular behavior. In the higher-rank
case, the corresponding formal vectors contain essential singularities in the
highest irregular parameter; for example, exponential factors of the form
\[
    \exp\left(\sum_{k>0} \frac{g_k}{c_r^k}\right).
\]
This distinction is important: the latter objects are not obtained by a
straightforward Whittaker construction and require a more delicate recursive
analysis.

Irregular vectors with genuinely irregular singular behavior were studied by
Gaiotto and Teschner \cite{GT} in connection with Argyres-Douglas theories.
They wrote down the expected Virasoro constraints for general rank and gave
explicit descriptions in the rank two and rank three cases, revealing recursive
structures that are expected to persist in general. A general form of the
defining relations for integer-rank irregular vectors was later written down by
Nishinaka and Uetoko \cite{NU19}. These works provided the physical and
algebraic motivation for the general higher-rank theory, but a mathematical
proof of the existence and uniqueness of the corresponding formal vectors was
not available.

Another approach to irregular conformal blocks is provided by irregular vertex
operators. In \cite{Nagoya15, Nagoya18}, such operators were introduced through
suitable intertwining relations involving irregular Verma modules and irregular
vectors, and they are used to construct irregular conformal blocks relevant to Painlev\'e
equations \cite{Nagoya15}, quantum Painlev\'e equations \cite{BST25}, and
confluent Heun equations \cite{INS26, LN22}. The recursive construction of
these operators already contains part of the mechanism used in the present
paper.

Half-integer ranks have appeared in the theory of irregular conformal blocks
since the early works on Gaiotto vectors \cite{Gaiotto} and irregular vertex
operators \cite{Nagoya18}. In the setting of genuinely irregular vectors,
the rank \(5/2\) case was proposed in \cite{PP}. A significant conjectural
framework for arbitrary half-integer ranks was later proposed by Hamachika,
Nakanishi, Nishinaka, and Tanigawa \cite{HNNT}, motivated by confluent limits
from integer-rank equations and supported by extensive finite-rank
computations. 
A related approach was developed in \cite{IILZ25}, where a first-order
differential realization of \(\mathrm{Vir}_+\) for arbitrary half-integer rank
was given in half-integer-indexed parameters, uniquely up to gauge equivalence.
Very recently, Zang gave a complementary construction using a
\(\mathbb Z_2\)-twisted boson and related it to the eigenvalue-recursion
parametrization; in that work the equivalence of the two parameterizations is
demonstrated explicitly in ranks \(3/2\) and \(5/2\).

The present paper adopts the parameter-space viewpoint of \cite{PP}: the
half-integer-rank vector is regarded as a vector in an integer-rank irregular
Verma module, with parameters \(c_0,c_1,\ldots,c_{r-1},\Lambda_{2r-1}\).
In this setting, however, a general construction of the Virasoro vector fields, as well as 
 an existence and uniqueness theorem for the corresponding formal vectors, was still missing.

The purpose of this paper is to provide such a construction. We prove the
existence and uniqueness of irregular vectors of arbitrary integer and
half-integer ranks. The key idea is to isolate the derivative with respect to
the highest irregular parameter: \(c_r\) in the integer-rank case and
\(\Lambda_{2r-1}\) in the half-integer-rank case.

Our construction is driven by a canonical operator \(L_*\).  It is obtained
from the coefficient matrix of the relevant parameter vector fields and
eliminates all intermediate deformation derivatives.  The result is a closed
recursive system for the coefficients of the formal irregular vector.  In the
half-integer-rank case, we first construct the required truncated Virasoro
vector fields intrinsically from their commutation relations with the scalar
higher modes; after passing to eigenvalue coordinates, these vector fields
recover the differential realizations appearing in the literature.

A related observation appears in \cite{IILZ25} in the rank \(5/2\) situation
related to the Painlev\'e I tau function, where a Virasoro generator
corresponding to differentiation with respect to the natural deformation
parameter is used to compute the expansion under an existence assumption.  The
operator \(L_*\) in the present paper gives a general higher-rank version of
this idea and is incorporated into the existence proof itself.

The higher-mode equations and the \(L_*\)-equation determine a canonical
solution.  We then prove that the lower Virasoro deformation equations are
compatible with this construction: their obstruction is scalar and can be
absorbed into a scalar gauge factor.  Thus, after scalar normalization, the
canonical solution satisfies the full lower-mode constraints.

These results give a rigorous foundation for irregular conformal blocks built
from higher-rank irregular vectors. They also clarify the algebraic deformation
structure underlying these formal objects.

The paper is organized as follows. 
In Section~2, we review the Virasoro algebra, irregular Verma modules,
and the construction of rank one irregular (Gaiotto) vectors.
We also recall a determinant formula from our previous work that will be used later.
Section~3 is devoted to the construction and uniqueness theorem for irregular
vectors of arbitrary integer rank. In particular, we introduce the canonical
operator \(L_*\) and derive the corresponding recursive system.
In Section~4, we construct the truncated Virasoro vector fields associated with
half-integer-rank singularities and prove their commutation relations.
Finally, we establish the existence and uniqueness theorem for
irregular vectors of arbitrary half-integer rank.

\section{Preliminaries}
The Virasoro algebra 
\begin{equation*}
	\Vir=\bigoplus_{n\in\Z}\C L_n\oplus \C c
\end{equation*}
is a Lie algebra with commutation relations
\begin{align*}
	&[L_m,L_n]=(m-n)L_{m+n}+\frac{c}{12}(m^3-m)\delta_{m+n,0}, 
\\ &[\Vir, c]=0.
\end{align*}
Let us denote by $M_\Delta$ a Verma module of the Virasoro algebra with the highest weight 
vector $|\Delta\rangle$. For a positive integer $r$ a vector $|I^{(r)}\rangle$ with 
$\Lambda=(\Lambda_r,\Lambda_{r+1},\ldots,\Lambda_{2r})\in \C^{r}\times \C^*$ satisfying 
\begin{equation}\label{eqdefIV1}
  L_n|I^{(r)}\rangle =\Lambda_n|I^{(r)}\rangle\quad (n=r,r+1,\ldots,2r)
\end{equation}
is called an irregular vector of rank $r$. 
A vector $|I^{(r-1/2)}\rangle$ with 
$\Lambda=(\Lambda_r,\Lambda_{r+1},\ldots,\Lambda_{2r-1})\in \C^{r-1}\times \C^*$ satisfying 
\begin{equation}\label{eqdefIV2}
  L_n|I^{(r-1/2)}\rangle =\Lambda_n|I^{(r-1/2)}\rangle\quad (n=r,r+1,\ldots,2r-1)
\end{equation}
is called an irregular vector of rank $r-1/2$. 
An irregular Verma module $M_\Lambda^{(r)}$ of rank $r$ is defined as an induced module
\begin{equation*}
    M_\Lambda^{(r)}= \mathrm{Ind}_{\mathrm{Vir}_{\geq r}} ^{\mathrm{Vir}} \C |I^{(r)}\rangle, 
\end{equation*}
 where $\mathrm{Vir}_{\geq r}$ ($r\in\Z_{\geq 0}$) is the subalgebra generated by $L_n$ ($n\geq r$).  
An irregular Verma module $M_\Lambda^{(r)}$ of rank $r$ is irreducible \cite{FJK12}.

For the determinant formula below, we put
\[
  \Lambda_n=0 \qquad (n>2r),
  \qquad
  \widetilde L_n:=L_n-\Lambda_n \qquad (n\ge r).
\]
Set 
\begin{equation*} 
\widetilde{L}_\lambda=\widetilde{L}_{\lambda_\ell+r}\cdots \widetilde{L}_{\lambda_1+r},\quad 
  L_{-\lambda}=L_{-\lambda_1+r}\cdots L_{-\lambda_\ell+r}
\end{equation*}
for a partition $\lambda=(\lambda_1,\ldots, \lambda_\ell)$. Then, $L_{-\lambda}|I^{(r)}\rangle$ for all partitions $\lambda$ forms the PBW basis of the irregular Verma module. We denote by $\{u\}$ the constant term ($|I^{(r)}\rangle$-component) for $u\in M_\Lambda^{(r)}$.  In the next sections, we use 
\begin{lem}[Corollary 2.23 \cite{Nagoya15}]\label{lem:det}
    For any $n,m\in\Z_{\geq 0}$, 
    \begin{equation*}
        \det \left(\left\{
        \widetilde{L}_\lambda L_{-\mu}|I^{(r)}\rangle\right\}\right)
        _{n\leq |\lambda|,|\mu|\leq m}=\Lambda_{2r}^{\sum_{i=n}^m ip(i)},
    \end{equation*}
    where $p(i)$ is the partition number of $i$. 
\end{lem}

\paragraph{Coefficient rings and formal completions.}
\label{par:coefficient-rings}
Throughout the paper the central element \(c\) acts by the scalar
\(1+6Q^2\), and we write
\[
  \Delta_a=a(Q-a).
\]
The central element \(c\) should not be confused with the parameters
\(c_i\).

We shall use the following coefficient rings and completions.  For the
integer-rank induction step \(r-1\to r\), where \(r\ge 2\), set
\[
  \mathcal K_r^{\mathrm{int}}
  :=\mathbb C(Q,c_0,c_0'),
  \qquad
  \mathcal A_r^{\mathrm{int}}
  :=
  \mathcal K_r^{\mathrm{int}}
  [c_1,\ldots,c_{r-1},c_{r-1}^{-1}].
\]
For the scalar-gauge discussion we also use the Laurent extension
\[
  \mathcal A_{r,\mathrm{Laur}}^{\mathrm{int}}
  :=
  \mathcal K_r^{\mathrm{int}}
  [c_1^{\pm1},\ldots,c_{r-1}^{\pm1}].
\]
Similarly, for the half-integer-rank construction \(r-1\to r-\frac12\),
set
\[
  \mathcal K_r^{\mathrm{half}}
  :=\mathbb C(Q,c_0),
  \qquad
  \mathcal A_r^{\mathrm{half}}
  :=
  \mathcal K_r^{\mathrm{half}}
  [c_1,\ldots,c_{r-1},c_{r-1}^{-1}],
\]
and
\[
  \mathcal A_{r,\mathrm{Laur}}^{\mathrm{half}}
  :=
  \mathcal K_r^{\mathrm{half}}
  [c_1^{\pm1},\ldots,c_{r-1}^{\pm1}].
\]
The smaller rings
\(\mathcal A_r^{\mathrm{int}}\) and
\(\mathcal A_r^{\mathrm{half}}\) are used for the recursive construction of
the canonical vectors.  The Laurent extensions are used only when scalar gauge
factors and logarithmic one-forms in the lower parameters are discussed.

Let \(M\) be the coefficient irregular Verma module which appears in a given
induction step, and let \(v_0\) denote its cyclic irregular vector.  For a
coefficient ring \(\mathcal A\) as above, we extend scalars to \(\mathcal A\)
and write
\[
  M_{\mathcal A}[d]
  :=
  \bigoplus_{|\lambda|=d}
  \mathcal A\, L_{-\lambda}v_0,
  \qquad
  \widehat M_{\mathcal A}
  :=
  \prod_{d\ge0} M_{\mathcal A}[d]. 
\]
  All formal vectors below are regarded as elements of this
PBW-degree completion.  At each fixed order in the highest irregular
parameter, the vectors constructed below contain only finitely many PBW degrees;
therefore Lemma~\ref{lem:det} is applied degree by degree.

We shall also fix the meaning of the formal exponential sectors used in the
construction.  Let \(t\) be the highest irregular parameter, and let
\(\mathcal A\) be one of the coefficient rings above.  For
\(g_1,\ldots,g_s,\nu\in \mathcal A\), we regard
\[
  t^\nu
  \exp\left(\sum_{j=1}^s \frac{g_j}{t^j}\right)
  \left(\widehat M_{\mathcal A}[[t]]\right)
\]
as a formal exponential sector.  The factor
\[
  P_t
  :=
  t^\nu
  \exp\left(\sum_{j=1}^s \frac{g_j}{t^j}\right)
\]
is treated as a formal prefactor.  Differentiation with respect to \(t\) is
defined by
\[
  \frac{\partial}{\partial t}\bigl(P_t F(t)\bigr)
  =
  P_t\left(
    \frac{\partial F}{\partial t}
    +
    \left(
      \frac{\nu}{t}
      -
      \sum_{j=1}^s \frac{j g_j}{t^{j+1}}
    \right)F(t)
  \right),
  \qquad
  F(t)\in \widehat M_{\mathcal A}[[t]].
\]
Here the coefficients \(g_j\) and \(\nu\) are independent of the highest
parameter \(t\), although they may depend on the lower parameters.

In the integer-rank construction we take
\[
  t=c_r,\qquad
  \mathcal A=\mathcal A_r^{\mathrm{int}},\qquad
  s=r-1,
\]
whereas in the half-integer-rank construction we take
\[
  t=\Lambda_{2r-1},\qquad
  \mathcal A=\mathcal A_r^{\mathrm{half}},\qquad
  s=r-1.
\]
Whenever we say that a scalar formal series in a prescribed exponential sector
is invertible, we mean that, after dividing by the fixed prefactor \(P_t\), its
constant term in \(\mathcal A[[t]]\) is a unit of \(\mathcal A\).

When $r=1$, an irregular vector 
can be constructed in terms of elements of $M_\Delta$. 
Set $\Lambda_1=2(Q-c_0)c_1$, $\Lambda_2=-c_1^2$, and 
$|I^{(1)}\rangle=\sum_{k=0}^\infty c_1^k v_k$ ($v_k\in M_\Delta$, $v_0=|\Delta\rangle$). 
Then, the relations \eqref{eqdefIV1} determine $v_k$ uniquely. Hence, 
an irregular vector of rank one is established as a power series in $c_1$ 
whose coefficients are elements of $M_\Delta$. 
We note that $L_0$ acts on the irregular vector as 
\begin{equation*}
L_0|I^{(1)}\rangle=\left(\Delta+c_1\frac{\partial}{\partial c_1}\right)|I^{(1)}\rangle.
\end{equation*}

 If $r>1$, then 
the relations \eqref{eqdefIV1} do not determine $v_k$ uniquely \cite{FJK12}. 
Instead, Gaiotto-Teschner \cite{GT} considered different series for $r=2$ ((3.5) in \cite{GT}) 
\begin{equation*}
  |I^{(2)}\rangle =c_1^{\nu_1}c_2^{\nu_2} e^{g_0(c_1)+g_1(c_1)/c_2}\sum_{k=0}^\infty c_2^k v_k.
\end{equation*}
Here, $v_0$ is the irregular vector $|I^{(1)}\rangle$ of rank one, and $v_k$ are its descendants. 
We can observe that the relations \eqref{eqdefIV1} with 
\begin{equation*}
  \Lambda_2=-c_1^2+c_2(3Q-2c_0),\quad \Lambda_3=-2c_1c_2,\quad \Lambda_4=-c_2^2
\end{equation*}
and 
\begin{equation*}
  L_0|\Lambda\rangle = c_0(Q-c_0)+c_1\frac{\partial}{\partial c_1}+2c_2\frac{\partial}{\partial c_2},\quad 
  L_1|\Lambda\rangle =2c_1(Q-c_0)+c_2\frac{\partial}{\partial c_1}
\end{equation*}
determine the parameters $\nu_1,\nu_2$, the functions $g_0(c_1)$, 
$g_1(c_1)$, and the descendants $v_k$. In the same approach, the irregular vector of rank $r$ 
could be constructed from the irregular vector of rank $r-1$ \cite{GT}, \cite{NU19}. 

We record the expected full form of the integer-rank irregular vector.  This
form will be established after scalar gauge normalization, in Section~3. We parametrize weights as 
\begin{equation*}
  \Lambda_{n}=((n+1)Q-2c_0)c_n-\sum_{k=1}^{n-1} c_k c_{n-k} \quad (n=1,\ldots,2r), 
\end{equation*}
where $c_i=0$ if $i<0$ or $i> r$. 

  For any \(r\in\Z_{\geq 2}\), the expected irregular vector is of the form
  \begin{equation}\label{eq:conjectural-integer-series}
    |I^{(r)}\rangle =c_1^{\nu_1}\cdots c_r^{\nu_r} \exp\left(\sum_{k=0}^{r-1} \frac{g_k(c_0',c_0,c_1,\ldots,c_{r-1})}{c_r^k}\right)
    \sum_{k=0}^\infty c_r^k v_k  
  \end{equation}
  satisfying 
  \begin{equation}\label{eq_irregular_condition_integer}
    L_n |I^{(r)}\rangle=\left(\Lambda_n+\sum_{k=1}^{r-n}k c_{n+k}\frac{\partial}{\partial c_k}\right) 
    |I^{(r)}\rangle\quad (n>0),\quad 
    L_0|I^{(r)}\rangle=\left(\Delta_{c_0}+\sum_{k=1}^rkc_k\frac{\partial}{\partial c_k}\right)
    |I^{(r)}\rangle,  
  \end{equation}
  where $v_0$ is an irregular vector $|I^{(r-1)}\rangle$ of rank $r-1$  satisfying the relations
  \begin{align*}
    L_n |I^{(r-1)}\rangle=\left(\Lambda_n'+\sum_{k=1}^{r-n}k c_{n+k}\frac{\partial}{\partial c_k}\right) 
    |I^{(r-1)}\rangle\quad (n>0),\quad 
    L_0|I^{(r-1)}\rangle=\left(\Delta_{c_0'}+\sum_{k=1}^{r-1}kc_k\frac{\partial}{\partial c_k}\right)
    |I^{(r-1)}\rangle, 
  \end{align*}
  where \begin{equation*}
    \Lambda_n'=((n+1)Q-2c_0')c_n-\sum_{k=1}^{n-1} c_k c_{n-k},  \quad c_j=0 \text{ for } j\geq r,  
\end{equation*} 
and $v_k$ are descendants of $v_0=|I^{(r-1)}\rangle$.

\begin{re}
The half-integer-rank case requires a small additional comment.  The
highest-mode eigenvalue conditions can be written down formally, but a
formal-series construction analogous to the integer-rank case also
involves the lower Virasoro actions, or equivalently the corresponding
Virasoro vector fields on the chosen parameter space.  In the present paper
we use the parameter space
\[
  c_0,c_1,\ldots,c_{r-1},\Lambda_{2r-1},
\]
and the required vector fields will be constructed in
Section~4.  We therefore postpone the
precise half-integer-rank formulation until then.
\end{re}

\section{Construction of the Irregular Vector of Integer Rank}

We establish a rigorous construction of the rank \(r\) irregular vector
\(|I^{(r)}\rangle\) from the rank \(r-1\) irregular vector
\(|I^{(r-1)}\rangle\), for \(r\in\mathbb Z_{\ge2}\).

\subsection{Construction of the Canonical Operator}
To uniquely construct the canonical irregular vector appearing in
Theorem~\ref{thm:int_to_int}, it is not necessary to use all the Virasoro
generators \(L_n\) with \(n\ge0\).
 The necessary ingredients are the higher generators $L_n$ ($n\geq r$) and an appropriate master operator $L_*$. 
We remark that for the specific case of rank $5/2$, an analogous operator was introduced in \cite{IILZ25} by requiring the existence of an anti-homomorphism from a certain Virasoro subalgebra to the ring of differential operators. 
Our viewpoint in this paper is based on a constructive motivation applicable to general rank $r$; we define $L_*$ systematically as a specific combination of $L_0, L_1, \ldots, L_{r-1}$ designed to directly eliminate the intermediate differential operators $\dfrac{\partial}{\partial c_i}$ ($i=1,\ldots,r-1$), thereby leading to a closed system of recurrence relations.

The elimination of the differential operators is performed as follows. 
Let the pure differential operator part obtained by subtracting the scalar part be denoted as 
\begin{equation*}
    \mathcal{D}_n = \sum_{k=1}^{r-n}k c_{n+k}\frac{\partial}{\partial c_k} \quad (n>0), \quad 
    \mathcal{D}_0 = \sum_{k=1}^rkc_k\frac{\partial}{\partial c_k}.
\end{equation*}

Writing these out for $n=0, 1, \dots, r-1$ and summarizing them in matrix notation $\mathcal{D} = M \partial_c$, we obtain the following:
\begin{equation}
    \begin{pmatrix}
        \mathcal{D}_0 \\
        \mathcal{D}_1 \\
        \vdots \\
        \mathcal{D}_{r-2} \\
        \mathcal{D}_{r-1}
    \end{pmatrix}
    =
    \begin{pmatrix}
        c_1 & 2c_2 & \cdots & (r-1)c_{r-1} & r c_r \\
        c_2 & 2c_3 & \cdots & (r-1)c_r & 0 \\
        \vdots & \vdots & \reflectbox{$\ddots$} & 0 & 0 \\
        c_{r-1} & 2c_r & \cdots & 0 & 0 \\
        c_r & 0 & \cdots & 0 & 0
    \end{pmatrix}
    \begin{pmatrix}
        \frac{\partial}{\partial c_1} \\
        \frac{\partial}{\partial c_2} \\
        \vdots \\
        \frac{\partial}{\partial c_{r-1}} \\
        \frac{\partial}{\partial c_r}
    \end{pmatrix}
    \label{eq:matrix_form}
\end{equation}

This $r \times r$ coefficient matrix $M$ has an anti-upper-triangular structure, where all elements below the anti-diagonal are zero. Hence, the determinant is \begin{equation}
    \det(M) = (-1)^{\frac{r(r-1)}{2}} r! (c_r)^r,
\end{equation}
thus, the matrix $M$ is invertible. Consequently, the inverse matrix $M^{-1}$ exists, and each partial derivative $\frac{\partial}{\partial c_k}$ can be expressed as a linear combination of $\mathcal{D}_0, \dots, \mathcal{D}_{r-1}$.

Let us particularly consider extracting only $\frac{\partial}{\partial c_r}$ (the $r$-th component of the column vector $\partial_c$). Then 
\begin{equation*}
    \frac{\partial}{\partial c_r} = \sum_{n=0}^{r-1} (M^{-1})_{r, n+1} \mathcal{D}_n.
\end{equation*}
We define the corresponding Virasoro operator by 
\begin{equation}\label{eq:def-Lstar-int}
    L_*:=c_r^r  \sum_{n=0}^{r-1} (M^{-1})_{r, n+1} D_n,
\end{equation}
where \(D_0=L_0-\Delta_{c_0}\) and \(D_n=L_n-\Lambda_n\) for \(1\le n\le r-1\). We note that since we multiply $c_r^r$ by $(M^{-1})_{r,n+1}$, this implies that $L_*$ is a polynomial of $c_1,\ldots,c_r$ of degree at most $r-1$ because all entries of the original matrix $M$ are linear polynomials in $c_k$ (or zero). 

\subsection{Canonical Irregular Vector of Integer Rank}
In the following theorem, we fix the scalar gauge by omitting the factor 
\(c_1^{\nu_1}\cdots c_{r-1}^{\nu_{r-1}}e^{g_0}\). 
Since the construction below only uses the higher-mode conditions and the \(L_*\)-equation, no derivatives with respect to the lower parameters occur. 
Thus multiplication by a non-zero scalar function of \(c_1,\ldots,c_{r-1}\) does not affect the existence and uniqueness statement proved here.
\begin{thm}\label{thm:int_to_int}
For any $r\in \mathbb{Z}_{\geq 2}$ there exists a unique formal series
\begin{equation}\label{eq:series_int_to_int}
    |\widetilde I^{(r)}\rangle = c_r^{\nu_r} \exp\left(\sum_{k=1}^{r-1} \frac{g_k(c_0',c_0,c_1,\ldots,c_{r-1})}{c_r^k}\right)
    \sum_{k=0}^\infty c_r^k v_k  
\end{equation}
satisfying 
\begin{equation}\label{eq:integer-canonical-conditions}
    L_*|\widetilde I^{(r)}\rangle= c_r^r\frac{\partial}{\partial c_r} |\widetilde I^{(r)}\rangle,\quad 
    L_n |\widetilde I^{(r)}\rangle=\Lambda_{n} 
    |\widetilde I^{(r)}\rangle\quad (n=r,r+1,\ldots,2r),
\end{equation}
where $v_0$ is an irregular vector $|I^{(r-1)}\rangle$ of rank $r-1$ satisfying the relations
\begin{align*}
    L_n |I^{(r-1)}\rangle=\Lambda_n'
    |I^{(r-1)}\rangle\quad (n=r-1,r,\ldots,2r-2),  
\end{align*}
where 
\begin{equation*}
    \Lambda_n'=((n+1)Q-2c_0')c_n-\sum_{k=1}^{n-1} c_k c_{n-k},  \quad c_j=0 \text{ for } j\geq r,  
\end{equation*}
and $v_k$ are descendants of $v_0=|I^{(r-1)}\rangle$.
\end{thm}
 
\begin{proof}
Set
\begin{equation*}
    L_*=\sum_{i=0}^{r-1}c_r^i f_i(L_0,L_1,\ldots,L_{r-1}). 
\end{equation*}
Here, $f_i$ are polynomials (operators) in $L_0, \dots, L_{r-1}$ with coefficients depending on $c_1, \dots, c_{r-1}$.

The defining relations \eqref{eq:integer-canonical-conditions} are equivalent to the following relations for $v_k$:
\begin{align}
    &\sum_{i=0}^{r-2} \left( f_i+(r-1-i)g_{r-1-i}    \right)  v_{k-i}
    +\left(f_{r-1}-\nu_r-k+r-1 \right)v_{k-r+1}=0,\label{eq:L_star_vk}
    \\
   & \widetilde{L}_r v_k = ((r+1)Q - 2c_0) v_{k-1}, \label{eq:L_r_vk}\\
    &\widetilde{L}_n v_k = -2 c_{n-r} v_{k-1} \quad (r < n < 2r), \label{eq:Ln-between-r-and-2r-vk}\\
    &L_{2r} v_k = - v_{k-2}, \label{eq:L_2r_vk}\\
    &L_n v_k = 0 \quad (n > 2r).\label{eq:Ln-greater-2r-vk}
\end{align}
We use the convention \(v_j=0\) for \(j<0\). 
This is because, in the definition of $\Lambda_n$, the constant term independent of $c_r$ exactly matches the weight $\Lambda_n'$ of rank $r-1$. Therefore, by comparing the coefficients of $c_r^k$ on both sides and setting $\widetilde{L}_n = L_n - \Lambda_n'$, we obtain the above relations.

It suffices to prove that the vectors $v_k$ satisfying these relations \eqref{eq:L_star_vk}--\eqref{eq:Ln-greater-2r-vk} can be uniquely constructed for all $k \in \mathbb{Z}_{\geq 1}$.

Set 
\begin{equation*}
  \widetilde{L}_\lambda=\widetilde{L}_{\lambda_\ell+r-1}\cdots \widetilde{L}_{\lambda_1+r-1},\quad 
  L_{-\lambda}=L_{-\lambda_1+r-1}\cdots L_{-\lambda_\ell+r-1}
\end{equation*}
for a partition $\lambda=(\lambda_1,\ldots, \lambda_\ell)$. 
Because of the relations \eqref{eq:L_r_vk}--\eqref{eq:Ln-greater-2r-vk}, we can denote $v_k$ by 
\begin{equation*}
  v_k=\sum_{|\lambda|\leq rk}c_\lambda^{(k)}L_{-\lambda}|I^{(r-1)}\rangle. 
\end{equation*}

First, we construct $v_k$ satisfying the relations \eqref{eq:L_r_vk},\eqref{eq:Ln-between-r-and-2r-vk}, \eqref{eq:L_2r_vk}, and \eqref{eq:Ln-greater-2r-vk} inductively. 
By definition, we have the matrix equation:
\begin{equation*}
 \left( \left\{ \widetilde{L}_\mu v_k \right\}\right)_{\mu}=
 \left( \left\{\widetilde{L}_\mu L_{-\lambda}|I^{(r-1)}\rangle\right\}\right)_{\mu,\lambda}
  \left(c_\lambda^{(k)}\right)_\lambda,  
\end{equation*}
where the summation is over all partitions $\lambda$, $\mu$ such that $1\leq |\lambda|, |\mu|\leq rk$. Lemma \ref{lem:det} is applied to the rank $r-1$ irregular module, hence the determinant of the associated Gram matrix is proportional to powers of 
  $\Lambda_{2r-2}'=-c_{r-1}^2$.  
Hence, 
 $c_\lambda^{(k)}$ for $|\lambda|\geq 1$
 are uniquely solved as polynomials in 
\begin{equation*}
    g_1,\ldots,g_{r-1},\nu_r,c_0',c_0,c_1,\ldots,c_{r-1},  c_{r-1}^{-1},c_\emptyset^{(1)},\ldots,c_\emptyset^{(k-1)}, 
\end{equation*}
  and the relations \eqref{eq:L_r_vk}--\eqref{eq:Ln-greater-2r-vk} hold. 

Next, we show that $g_i$ ($i=1,\ldots,r-1$), $\nu_r$, and $c_\emptyset^{(i)}$ ($i\geq 1$) can be uniquely solved as polynomials in 
\begin{equation*}
    c_0',c_0,c_1,\ldots,c_{r-1},c_{r-1}^{-1}.
\end{equation*} 

Setting $k=0$ in equation (\ref{eq:L_star_vk}), we obtain
\begin{equation}
    f_0v_0 = - (r-1) g_{r-1} v_0.
\end{equation}
From Lemma \ref{lem:f0}, $f_0$ depends only on $L_{r-1}$, hence the unknown function $g_{r-1}$ is uniquely determined.

We proceed by mathematical induction. Suppose that $v_0,v_1,\ldots,v_{k-1}$ satisfy the relations \eqref{eq:L_star_vk}--\eqref{eq:Ln-greater-2r-vk}, and that $c_\lambda^{(j)}$ ($j\leq k-1$) are polynomials in 
\begin{equation*}
    c_0',c_0,c_1,\ldots,c_{r-1},c_{r-1}^{-1},c_\emptyset^{(k-1)},\ldots,
    c_\emptyset^{(k-r+1)}.
\end{equation*} 
Furthermore, for $k\leq r$, assume that $g_{r-1},\ldots,g_{r-k}$ are polynomials in $c_0',c_0,c_1,\ldots,c_{r-1},c_{r-1}^{-1}$. 
For $k>r$, assume that 
\begin{equation*}
    g_{r-1},\ldots,g_1,\nu_r,c_\emptyset^{(1)},\ldots,c_\emptyset^{(k-r)}
\end{equation*}
are polynomials in $c_0',c_0,c_1,\ldots,c_{r-1},c_{r-1}^{-1}$. 

  We set $X_0=v_0$, for $k>0$
  \[
 X_k
 =
 v_k-\sum_{i=\max(1,k-r+1)}^k
 c_\emptyset^{(i)}X_{k-i},
\]
and $X_k=0$ for $k<0$. 
   For $k>0$, $X_k$ has no constant term. Furthermore, it does not depend on $c_\emptyset^{(j)}$. For example, since \(X_1=v_1-c_\emptyset^{(1)}v_0\), applying
\(\widetilde L_n\) to \(X_1\) produces expressions involving only the
parameters \(c_0',c_0,c_1,\ldots,c_{r-1}\).  The same property follows for all
\(X_k\) by induction.
Put
\begin{align*}
&A=f_0+(r-1)g_{r-1},\quad
B_j=-\left(f_j+(r-1-j)g_{r-1-j}\right)
\quad (1\leq j\leq r-2),\\
&C_k=-f_{r-1}+\nu_r+k-r+1.    
\end{align*}
Applying \(A\) to \(X_k\) and using the \(L_*\)-recurrence gives 
\begin{align*}
A X_k
={}&
\sum_{i=1}^{r-2} B_i X_{k-i}
+
\sum_{i=1}^{r-2}
\sum_{j=\max(1,k-i-r+1)}^{k-r}
c_{\emptyset}^{(j)}B_iX_{k-i-j}  
+
C_k v_{k-r+1}
-
c_\emptyset^{(k-r+1)} C_{r-1} v_0 .
\end{align*}
Here, the second sum is understood to be empty when the upper bound is smaller than the lower bound, and $c_\emptyset^{(j)}=0$ for ($j\leq 0$) in this formula.

Taking the constant term determines \(g_{r-k-1}\) if \(k<r-1\),
\(\nu_r\) if \(k=r-1\), and \(c_\emptyset^{(k-r+1)}\) if \(k>r-1\),
successively, as polynomials in
\(c'_0,c_0,\ldots,c_{r-1},c_{r-1}^{-1}\).
Indeed, in the last case, \(c_\emptyset^{(j)}\) for \(j\leq k-r\)
have already been determined as polynomials in
\(c'_0,c_0,\ldots,c_{r-1},c_{r-1}^{-1}\), and the new coefficient appears
with coefficient
\[
C_k-C_{r-1}=k-r+1.
\]

Finally, we show that the relation \eqref{eq:L_star_vk} holds. 
Since $g_i$ ($i=1,\ldots,r-1$), $\nu_r$, $c_\emptyset^{(i)}$ ($i\geq 1$) are solved by looking at the constant term of \eqref{eq:L_star_vk}, 
we need to show that the actions of $\widetilde{L}_n$ ($n\geq r$) on the left-hand side of \eqref{eq:L_star_vk} are equal to zero. We again use induction on $k$, namely, we suppose that $v_0,v_1,\ldots,v_{k-1}$ satisfy the relations \eqref{eq:L_star_vk}.  Applying $\widetilde{L}_n$ on the left-hand side of \eqref{eq:L_star_vk}, we obtain
\begin{align}
    &\sum_{i=0}^{r-1}  [L_n,f_i]   v_{k-i}
    +2c_{n-r}v_{k-r}\quad (r<n<2r),
    \label{eq:Ln-act-Lstar}
    \\
    &\sum_{i=0}^{r-1}  [L_r,f_i]   v_{k-i}
    -((r+1)Q-2c_0)v_{k-r},\label{eq:Lr-act-Lstar}
    \\
   & \sum_{i=0}^{r-1}  [L_{2r},f_i]   v_{k-i}
    +2v_{k-r-1}\label{eq:L2r-act-Lstar}
\end{align}
by induction. We note that each $f_i$ does not depend on $k$. 

Noting that equation \eqref{eq:L_star_vk} is obtained by applying the $r$-th row of $M^{-1}(D_0,\ldots,D_{r-1})^\mathsf{T}$ to $|\widetilde I^{(r)}\rangle$, we compute the action as follows, rather than calculating the action on each term of \eqref{eq:L_star_vk} individually:
\begin{equation*}
    \widetilde{L}_n M^{-1}\left(L_0,\ldots,L_{r-1}\right)^\mathsf{T} |\widetilde I^{(r)}\rangle
    = M^{-1}\left( n\Lambda_n,(n-1)\Lambda_{n+1},\ldots,(n-r+1)\Lambda_{n+r-1}\right)^\mathsf{T} |\widetilde I^{(r)}\rangle. 
\end{equation*}
Since \(L_*\) was defined using the \(r\)-th row of
\(M^{-1}(D_0,\ldots,D_{r-1})^\mathsf T\), we compute
\begin{equation*}
    E(n)=\sum_{j=1}^r (M^{-1})_{r,j} (n-j+1) \Lambda_{n+j-1}.
\end{equation*}
Here, noting that
\[
  \Lambda_n
  =
  -\sum_{\substack{p+q=n\\ p,q>0}}c_pc_q
  \qquad (n\ge r+1)
\]
and
\[
  \Lambda_r
  =
  ((r+1)Q-2c_0)c_r
  -
  \sum_{\substack{p+q=r\\ p,q>0}}c_pc_q,
\]
where \(c_i=0\) for \(i>r\), for \(n>r\) we have
\begin{align*}
    E(n) =& -\sum_{j=1}^r (M^{-1})_{r,j} \sum_{\substack{p+q=n+j-1\\p,q>0}} (n-j+1) c_p c_q
    \\
=&-\sum_{j=1}^r (M^{-1})_{r,j} \sum_{\substack{p+q=n+j-1\\p,q>0}} (n-p+n-q) c_p c_q
 \\
 =& -2 \sum_p c_p \left( \sum_{j=1}^r (M^{-1})_{r,j} (n-p) c_{(n-p)+j-1} \right) \\
    =& -2 \sum_p c_p \delta_{r, n-p} \\
    =& -2 c_{n-r}.
\end{align*}
For $n=r$, from Lemma \ref{lem:inverse_A}, we have
\begin{equation*}
    (M^{-1})_{r,1}=\frac{1}{rc_r},
\end{equation*}
hence, since the quadratic terms vanish as above,   
\begin{equation*}
    E(r)=(r+1)Q-2c_0.
\end{equation*}
Therefore, since the definition of $L_*$ \eqref{eq:def-Lstar-int} and equation \eqref{eq:L_star_vk} are obtained by examining the coefficient of $c_r^k$ in $\left(L_*-c_r^r\dfrac{\partial}{\partial c_r}\right)|\widetilde I^{(r)}\rangle$, the expressions \eqref{eq:Ln-act-Lstar}, \eqref{eq:Lr-act-Lstar}, \eqref{eq:L2r-act-Lstar} are identically equal to zero. Denote by $Z_k$ the left-hand side of \eqref{eq:L_star_vk}. 
By construction, the coefficient of $v_0$ in $Z_k$ is zero. Moreover, the above computation shows that $\widetilde{L}_\mu Z_k=0$ for all relevant partitions $\mu$. Hence, by Lemma \ref{lem:det}, the non-constant descendant part of $Z_k$ is also zero. Therefore, $Z_k=0$. 
This completes the proof. 
\end{proof}

\subsection{Scalar Gauge and the Full Lower Virasoro Deformation Equations}

In this subsection, we prove that the canonical vector constructed in
Theorem~\ref{thm:int_to_int} can be normalized by a scalar gauge factor so that it
satisfies the full lower Virasoro deformation equations.

For the integer-rank parameter space, set
\[
 \mathcal L_0:=\Delta_{c_0}+\mathcal D_0,\qquad
 \mathcal L_n:=\Lambda_n+\mathcal D_n\quad(1\le n\le r-1),
\]
and, for \(n\ge r\),
\[
 \mathcal L_n:=\Lambda_n .
\]
We then define
\[
 R_n:=L_n-\mathcal L_n\qquad(n\ge0).
\]
Equivalently, for \(0\le n\le r-1\),
\[
 R_n=D_n-\mathcal D_n,
\]
while for \(n\ge r\),
\[
 R_n=L_n-\Lambda_n.
\]

\begin{lem}[Commutation with parameter vector fields]\label{lem:commutativity}
Assume that the rank \(r-1\) vector
\[
 v_0=|I^{(r-1)}\rangle
\]
satisfies the full lower Virasoro deformation equations
\[
 R_i^{(r-1)}v_0=0\qquad(0\le i\le r-2).
\]
Then the parameter vector fields \(\mathcal D_i\) which appear in the
rank \(r\) construction commute with the Virasoro action on the completed
rank \(r-1\) irregular Verma module:
\[
 [L_n,\mathcal D_i]=0
 \qquad(n\in\mathbb Z,\ 0\le i\le r-1).
\]
\end{lem}

\begin{proof}
The \(\partial_{c_r}\)-part of \(\mathcal D_i\) acts trivially on
\(v_0\) and on the completed rank \(r-1\) module, so it plainly commutes
with the Virasoro action.  The remaining part is a vector field on the
rank \(r-1\) parameter space.  Since the rank \(r-1\) vector fields
\[
 \mathcal D_0^{(r-1)},\ldots,\mathcal D_{r-2}^{(r-1)}
\]
form a frame after localizing at \(c_{r-1}\), it is enough to prove the
claim for these vector fields.

Fix \(0\le i\le r-2\).  We first check the commutation on the cyclic
vector \(v_0\).  If \(n<r-1\), then \(L_n\) is one of the PBW generators
of the rank \(r-1\) irregular Verma module, and the parameter vector
fields commute with it by definition.  Hence assume \(n\ge r-1\).

By the rank \(r-1\) lower equations,
\[
 \mathcal D_i^{(r-1)}v_0
 =
 \bigl(L_i-\Lambda_i^{(r-1)}\bigr)v_0,
\]
where for \(i=0\) we read \(\Lambda_0^{(r-1)}\) as the conformal weight.
On the other hand,
\[
 L_nv_0=\Lambda_n^{(r-1)}v_0
 \qquad(n\ge r-1),
\]
with the convention \(\Lambda_n^{(r-1)}=0\) for \(n>2r-2\).  Therefore
\[
\begin{aligned}
 L_n\mathcal D_i^{(r-1)}v_0
 &=(n-i)\Lambda_{n+i}^{(r-1)}v_0
   +\Lambda_n^{(r-1)}\mathcal D_i^{(r-1)}v_0 .
\end{aligned}
\]
The scalar parts satisfy
\[
 \mathcal D_i^{(r-1)}\Lambda_n^{(r-1)}
 =
 (n-i)\Lambda_{n+i}^{(r-1)}.
\]
Hence
\[
\begin{aligned}
 \mathcal D_i^{(r-1)}L_nv_0
 &=
 \mathcal D_i^{(r-1)}
 \bigl(\Lambda_n^{(r-1)}v_0\bigr)\\
 &=
 (n-i)\Lambda_{n+i}^{(r-1)}v_0
 +\Lambda_n^{(r-1)}\mathcal D_i^{(r-1)}v_0.
\end{aligned}
\]
Thus
\[
 [L_n,\mathcal D_i^{(r-1)}]v_0=0.
\]

It remains to extend this from \(v_0\) to all PBW descendants.  We argue
by induction on the PBW length.  Suppose that, for a vector \(u\), one
has
\[
 [L_m,\mathcal D_i^{(r-1)}]u=0
 \qquad\text{for all }m.
\]
Let \(L_p\) be a PBW generator, so \(p<r-1\).  Since
\(\mathcal D_i^{(r-1)}\) commutes with PBW generators by definition, we
have
\[
\begin{aligned}
 [L_n,\mathcal D_i^{(r-1)}]L_pu
 &=
 [L_n,L_p]\mathcal D_i^{(r-1)}u
 -\mathcal D_i^{(r-1)}\bigl([L_n,L_p]u\bigr)\\
 &\quad
 +L_p[L_n,\mathcal D_i^{(r-1)}]u .
\end{aligned}
\]
The last term is zero by the induction hypothesis.  The commutator
\([L_n,L_p]\) is a linear combination of \(L_{n+p}\) and, possibly, a
central scalar.  By the induction hypothesis again, \(L_{n+p}\) also
commutes with \(\mathcal D_i^{(r-1)}\) on \(u\), and the central term
commutes trivially.  Hence the first two terms cancel.  Therefore
\[
 [L_n,\mathcal D_i^{(r-1)}]L_pu=0.
\]
This proves the induction step, and hence
\[
 [L_n,\mathcal D_i^{(r-1)}]=0
\]
on the completed rank \(r-1\) irregular Verma module.

Since any rank \(r\) vector field \(\mathcal D_i\), restricted to the
rank \(r-1\) parameter space, is a scalar linear combination of
\(\mathcal D_0^{(r-1)},\ldots,\mathcal D_{r-2}^{(r-1)}\), and since
\(L_n\) does not act on scalar coefficients, the same commutation holds
for \(\mathcal D_i\).  This completes the proof.
\end{proof}

\begin{lem}
For \(m,n\ge0\),
\[
 [R_m,R_n]=(m-n)R_{m+n}.
\]
\end{lem}

\begin{proof}
By the preceding lemma, the Virasoro operators commute with the
parameter vector fields.  The scalar vector fields satisfy
\[
 [\mathcal L_m,\mathcal L_n]=(n-m)\mathcal L_{m+n}.
\]
Together with
\[
 [L_m,L_n]=(m-n)L_{m+n},
\]
we obtain
\[
 [R_m,R_n]=(m-n)R_{m+n}.
\]
\end{proof}

\begin{lem}
Let
\[
 W:=|\widetilde I^{(r)}\rangle
\]
be the canonical vector constructed in Theorem~\ref{thm:int_to_int}, and put
\[
 \Theta:=\{W\}.
\]
Then \(\Theta\) is a non-zero invertible scalar formal series in the
prescribed exponential sector over the coefficient ring
\(\mathcal A_r^{\mathrm{int}}\).
  If \(X\) is a formal vector satisfying
\[
 R_nX=0\qquad(n\ge r),
\]
then
\[
 X=\frac{\{X\}}{\Theta}\,W.
\]
Equivalently, within the same prescribed exponential sector and over the same
scalar coefficient ring, the space of solutions of the higher-mode equations
\(R_nX=0\) \((n\ge r)\) is a free rank-one module generated by \(W\).
\end{lem}

\begin{proof}
Indeed, in the prescribed exponential sector of Theorem~\ref{thm:int_to_int}, the
\(v_0\)-component of \(W\) has a non-zero leading term; hence it is
invertible in the localized scalar coefficient ring used throughout this
construction.

Set
\[
 a:=\frac{\{X\}}{\Theta}.
\]
Then \(aW\) also satisfies the higher-mode equations, and
\[
 \{X-aW\}=0.
\]
Thus \(X-aW\) is a solution of the higher-mode equations with zero
\(v_0\)-component.

We claim that such a solution must vanish.  Indeed, write \(X-aW\) in
the PBW basis
\[
 L_{-\lambda}v_0.
\]
After factoring out the common exponential and power prefactor, at each
fixed order in \(c_r\), the equations
\(R_nX=0\) for \(n\ge r\) give a linear system for the non-constant PBW
coefficients.  The coefficient matrix is the Gram matrix appearing in
Lemma~\ref{lem:det}, applied to the rank \(r-1\) irregular module.  Its determinant
is non-zero.  Hence all non-constant PBW coefficients are uniquely
determined by the \(v_0\)-component.  Since that component is zero, all
coefficients vanish.  Therefore
\[
 X-aW=0.
\]
This proves
\[
 X=\frac{\{X\}}{\Theta}W.
\]
\end{proof}

\begin{thm}\label{thm:integer_full_vector}
Assume inductively that the initial rank \(r-1\) vector
\(v_0=|I^{(r-1)}\rangle\) has been normalized to satisfy its full lower
Virasoro deformation equations.
Let \(W=|\widetilde I^{(r)}\rangle\) be the canonical integer-rank
irregular vector constructed in Theorem~\ref{thm:int_to_int}.  Then there exists a
non-zero scalar formal function
\[
 f=f(c_0',c_0,c_1,\ldots,c_{r-1}),
\]
unique up to multiplication by a non-zero function of \(c_0'\) and
\(c_0\), such that
\[
 |I^{(r)}\rangle:=fW
\]
satisfies
\[
 R_i|I^{(r)}\rangle=0,
 \qquad 0\le i\le r-1.
\]
\end{thm}

\begin{proof}
We shall use the fact that the canonical vector satisfies
\[
 R_nW=0\qquad(n\ge r),
\]
where \(\Lambda_n=0\) for \(n>2r\). 

Put
\[
 \Theta:=\{W\}.
\]
We first show that the obstruction to the lower equations is scalar.
For \(0\le i\le r-1\), set
\[
 Y_i:=R_iW.
\]
Let \(n\ge r\).  Since \(W\) satisfies the higher-mode equations,
\[
 R_nY_i=[R_n,R_i]W=(n-i)R_{n+i}W=0\qquad(n\ge r).
\]
Thus \(Y_i\) satisfies the higher-mode equations.  By the preceding
lemma,
\[
 Y_i=\frac{\{Y_i\}}{\Theta}W.
\]
Therefore there exists a scalar formal function
\[
 a_i:=\frac{\{R_iW\}}{\Theta}
\]
such that
\[
 R_iW=a_iW,
 \qquad 0\le i\le r-1.
\]

We next derive the integrability condition for the functions \(a_i\).
For \(0\le i,j\le r-1\), applying the Maurer--Cartan relation to \(W\)
gives
\[
 [R_i,R_j]W=(i-j)R_{i+j}W.
\]
We use the convention
\[
 a_k=0\qquad(k\ge r),
\]
since \(R_kW=0\) for \(k\ge r\). 
We obtain
\[
 \mathcal D_i a_j-\mathcal D_j a_i=(j-i)a_{i+j}.
\]
This is precisely the Frobenius integrability condition for
\[
 \mathcal D_i h=a_i,
 \qquad 0\le i\le r-1.
\]
Since the coefficient matrix of
\(\mathcal D_0,\ldots,\mathcal D_{r-1}\) with respect to
\(\partial_{c_1},\ldots,\partial_{c_r}\) has determinant
\[
(-1)^{r(r-1)/2}r!c_r^r,
\]
these vector fields form a frame after localizing at \(c_r\).
Therefore,
by the formal Frobenius theorem, there exists a scalar formal function
\[
 h=h(c_0',c_0,c_1,\ldots,c_r),
\]
unique up to a function independent of \(c_1,\ldots,c_r\), such that
\[
 \mathcal D_i h=a_i,
 \qquad 0\le i\le r-1.
\]

It remains to show that \(h\) is independent of \(c_r\).  Recall that
\(L_*\) was defined by extracting \(\partial/\partial c_r\) from the
 vector fields \(\mathcal D_0,\ldots,\mathcal D_{r-1}\). By construction of \(L_*\), we have
\[
 \frac{\partial}{\partial c_r}
 =
 \sum_{i=0}^{r-1}(M^{-1})_{r,i+1}\mathcal D_i
\]
and hence
\[
 L_*-c_r^r\frac{\partial}{\partial c_r}
 =
 c_r^r\sum_{i=0}^{r-1}(M^{-1})_{r,i+1}R_i.
\]
Since the canonical vector \(W\) satisfies the \(L_*\)-equation, 
\[
 0=
 \left(
 L_*-c_r^r\frac{\partial}{\partial c_r}
 \right)W
 =
 c_r^r\sum_{i=0}^{r-1}(M^{-1})_{r,i+1}R_iW.
\]
Using \(R_iW=a_iW\), we get
\[
 c_r^r\sum_{i=0}^{r-1}(M^{-1})_{r,i+1}a_i=0.
\]
Since \(\mathcal D_i h=a_i\), the left-hand side is exactly
\[
 c_r^r\frac{\partial h}{\partial c_r}.
\]
Therefore
\(
 h=h(c_0',c_0,c_1,\ldots,c_{r-1}).
\)
Notice that the functions \(a_i\) may depend on \(c_r\).  What the
\(L_*\)-equation proves is not the \(c_r\)-independence of the \(a_i\),
but the \(c_r\)-independence of the potential \(h\).

Set
\[
 f:=e^h.
\]
Then, 
\[
\begin{aligned}
 R_i(fW)
 &=fR_iW-(\mathcal D_i f)W \\
 &=fa_iW-fa_iW \\
 &=0.
\end{aligned}
\]
Thus
\[
 R_i|I^{(r)}\rangle=0,
 \qquad 0\le i\le r-1,
\]
where
\[
 |I^{(r)}\rangle=fW.
\]
Together with the rank-one construction, this proves the scalar-gauge
statement for all integer ranks by induction.

Finally, suppose \(f_1W\) and \(f_2W\) both satisfy the lower equations.
Then \(g=f_2/f_1\) satisfies
\[
 \mathcal D_i g=0,
 \qquad 0\le i\le r-1.
\]
Since \(\mathcal D_0,\ldots,\mathcal D_{r-1}\) form a frame after localizing at
\(c_r\), this implies that \(g\) is independent of
\(c_1,\ldots,c_r\).  Hence \(g\) may depend only on \(c_0'\) and \(c_0\).
This proves the asserted uniqueness.
\end{proof}

\begin{cor}
The scalar gauge factor \(f\) in the preceding theorem can be written as
\[
 f
 =
 C(c_0',c_0)\,
 c_1^{\nu_1}\cdots c_{r-1}^{\nu_{r-1}}
 \exp\bigl(g_0(c_0',c_0,c_1,\ldots,c_{r-1})\bigr),
\]
where \(g_0\) belongs to the same formal Laurent coefficient ring in the
lower parameters, and the exponents \(\nu_j\) depend only on the passive
parameters \(c_0'\) and \(c_0\).
\end{cor}

\begin{proof}
Let \(h=\log f\).  By the theorem, \(h\) is independent of the highest
parameter \(c_r\).  The equations
\[
 \mathcal D_i h=a_i
\]
therefore define a closed formal Laurent \(1\)-form in
\(c_1,\ldots,c_{r-1}\).  On the Laurent torus with coordinates
\(c_1,\ldots,c_{r-1}\), every closed formal Laurent \(1\)-form has the
decomposition
\[
 dh
 =
 dg_0+\sum_{j=1}^{r-1}\nu_j\,\frac{dc_j}{c_j}.
\]
Indeed, writing
\[
 dh=\sum_{j=1}^{r-1} A_j(c)\,d\log c_j
\]
and expanding \(A_j(c)=\sum_\alpha A_{j,\alpha}c^\alpha\), the closedness
condition gives
\[
 \alpha_i A_{j,\alpha}=\alpha_j A_{i,\alpha}.
\]
For \(\alpha\neq0\), this implies
\[
 A_{j,\alpha}=\alpha_j B_\alpha,
\]
so the corresponding term is
\[
 d(B_\alpha c^\alpha).
\]
Only the zero mode \(\alpha=0\) remains non-exact, giving the logarithmic
terms
\[
 \sum_{j=1}^{r-1}\nu_j d\log c_j.
\]
The coefficients \(\nu_j\) are the zero-mode Laurent coefficients, hence they are independent of \(c_1,\ldots,c_{r-1}\) and may depend only on the
passive parameters.  Thus
\[
 h
 =
 g_0+\sum_{j=1}^{r-1}\nu_j\log c_j+\log C(c_0',c_0).
\]
Exponentiating gives the desired expression for \(f\).
\end{proof}

\section{Construction from Integer to Half-Integer Rank}

As in the integer-rank case, we establish a rigorous construction of the
half-integer-rank irregular vector \(|I^{(r-1/2)}\rangle\) from the
integer-rank vector \(|I^{(r-1)}\rangle\), for \(r\in\mathbb Z_{\ge2}\).  We seek a formal vector expanded in terms of the highest half-integer parameter $\Lambda_{2r-1}$.

\subsection{Existence of the Truncated Virasoro Vector Fields for Half-Integer Rank}

Before discussing the construction of the irregular vectors, we must first establish that the differential representation of the Virasoro algebra on the parameter space of the half-integer-rank singularity is mathematically well-defined.

The construction below is intrinsic to the parameter space
\((c_1,\ldots,c_{r-1},\Lambda_{2r-1})\); after passing to eigenvalue
coordinates, it reproduces the vector-field part of the differential
realizations used in the half-integer-rank literature.
\begin{thm}\label{thm:half_vector_fields}
Let \(c_1,\ldots,c_{r-1}\) and \(\Lambda_{2r-1}\) be independent variables, and set
\[
S_m=-\sum_{a=1}^{r-1}c_a c_{m-a}\quad (r\le m\le 2r-2),
\qquad
S_{2r-1}=\Lambda_{2r-1},
\qquad
S_m=0\quad (m>2r-1),
\]
where \(c_a=0\) for \(a<1\) or \(a>r-1\).  Let
\[
\mathcal V_0
=
\sum_{k=1}^{r-1} k c_k\frac{\partial}{\partial c_k}
+(2r-1)\Lambda_{2r-1}\frac{\partial}{\partial \Lambda_{2r-1}} .
\]
Then, for each \(1\le n\le r-1\), there exists a unique rational vector field of the form
\[
\mathcal V_n
=
\sum_{k=1}^{r-n}
h_{n,k}(c_1,\ldots,c_{r-1},\Lambda_{2r-1})
\frac{\partial}{\partial c_k}
\]
such that
\[
\mathcal V_n(S_m)=(m-n)S_{m+n}
\qquad
(r\le m\le 2r-1).
\]
Moreover, these vector fields satisfy
\[
[\mathcal V_0,\mathcal V_n]=n\mathcal V_n,
\]
and, for \(1\le m,n\le r-1\),
\[
[\mathcal V_m,\mathcal V_n]
=
\begin{cases}
(n-m)\mathcal V_{m+n}, & m+n\le r-1,\\
0, & m+n\ge r.
\end{cases}
\]
Equivalently, if we put \(\mathcal V_N=0\) for \(N\ge r\), then
\[
[\mathcal V_m,\mathcal V_n]=(n-m)\mathcal V_{m+n}.
\]
\end{thm}

\begin{proof}
Put
\[
y_j=S_{r-1+j}\qquad (1\le j\le r),
\]
so that \(y_r=\Lambda_{2r-1}\).  We first observe that
\[
y_1,\ldots,y_r
\]
form a rational coordinate system, after localizing at \(c_{r-1}\).  Indeed, with respect to the coordinates
\[
x=(c_1,\ldots,c_{r-1},\Lambda_{2r-1}),
\]
the Jacobian matrix \(\left(\partial y_j/\partial x_k\right)\) is triangular.  For \(1\le j\le r-1\) and \(1\le k\le r-1\), we have
\[
\frac{\partial y_j}{\partial c_k}
=
\frac{\partial S_{r-1+j}}{\partial c_k}
=
-2c_{r-1+j-k}.
\]
This entry is zero when \(j>k\), and the diagonal entries are all equal to
\[
-2c_{r-1}.
\]
The last row is given by
\[
\frac{\partial y_r}{\partial \Lambda_{2r-1}}=1,
\]
and the other entries in the last row vanish.  Hence
\[
\det\left(\frac{\partial y_j}{\partial x_k}\right)
=
(-2c_{r-1})^{r-1}\neq 0.
\]
Therefore a rational derivation is uniquely determined by its action on
\(y_1,\ldots,y_r\).

For \(1\le n\le r-1\), prescribe a derivation by
\[
\mathcal V_n(y_j)
=
(r-1+j-n)y_{j+n},
\]
where we set \(y_\ell=0\) for \(\ell>r\).  Since \(y_1,\ldots,y_r\) form a rational coordinate system, this prescription determines a unique rational vector field.

It remains to check that this vector field has the asserted triangular form.  The equation
\[
\mathcal V_n(y_j)=(r-1+j-n)y_{j+n}
\]
for \(j=1,\ldots,r-n\) gives the linear system
\begin{equation}\label{eq:linear_sys_r_minus_1}
\sum_{k=1}^{r-n}
h_{n,k}
\frac{\partial S_{r-1+j}}{\partial c_k}
=
(r-1+j-n)S_{r-1+j+n}
\qquad
(1\le j\le r-n).
\end{equation}
Its coefficient matrix is
\[
A_{j,k}=-2c_{r-1+j-k}.
\]
This matrix is upper triangular, with diagonal entries \(-2c_{r-1}\).  Hence it is invertible, and the functions \(h_{n,k}\) are uniquely determined as rational functions.  The equations with \(j>r-n\) force no additional terms, because the right-hand side is zero and the triangular form gives no derivatives with respect to \(c_k\) for \(k>r-n\).  Also \(\mathcal V_n(y_r)=0\), so no \(\partial/\partial\Lambda_{2r-1}\) term occurs.

We now verify the commutation relations.  Since \(\mathcal V_0\) is the Euler vector field and \(y_j=S_{r-1+j}\) has weight \(r-1+j\), we have
\[
\mathcal V_0(y_j)=(r-1+j)y_j.
\]
Thus
\[
[\mathcal V_0,\mathcal V_n](y_j)
=
\mathcal V_0\bigl((r-1+j-n)y_{j+n}\bigr)
-
\mathcal V_n\bigl((r-1+j)y_j\bigr)
=
n\,\mathcal V_n(y_j).
\]
Since the \(y_j\)'s form a coordinate system, this proves
\[
[\mathcal V_0,\mathcal V_n]=n\mathcal V_n.
\]

Similarly, for \(1\le m,n\le r-1\), we compute on \(y_j\):
\[
\begin{aligned}
[\mathcal V_m,\mathcal V_n](y_j)
&=
(r-1+j-n)(r-1+j+n-m)y_{j+m+n} \\
&\quad
-
(r-1+j-m)(r-1+j+m-n)y_{j+m+n}  \\
&=
(n-m)(r-1+j-m-n)y_{j+m+n}.
\end{aligned}
\]
If \(m+n\le r-1\), the last expression is exactly
\[
(n-m)\mathcal V_{m+n}(y_j).
\]
If \(m+n\ge r\), then \(j+m+n>r\) for all \(j\ge1\), and hence the expression is zero.  Again, since the \(y_j\)'s form a coordinate system, the desired commutation relations follow.
\end{proof}

\begin{cor}[Eigenvalue coordinates]\label{cor:half_eigenvalue_coordinates}
On the open set \(c_{r-1}\neq 0\), the variables
\[
  \Lambda_r,\ldots,\Lambda_{2r-1}
\]
form a rational coordinate system equivalent to
\[
  c_1,\ldots,c_{r-1},\Lambda_{2r-1},
\]
where
\[
  \Lambda_m
  =
  -\sum_{a=1}^{r-1}c_a c_{m-a}
  \qquad (r\le m\le 2r-2). 
\]
In these coordinates, the vector fields of
Theorem~\ref{thm:half_vector_fields} are given by
\begin{equation}\label{eq:vector_field_Lambda}
  \mathcal V_n
  =
  \sum_{\ell=r}^{2r-1-n}
  (\ell-n)\Lambda_{\ell+n}
  \frac{\partial}{\partial\Lambda_\ell},
  \qquad
  0\le n\le r-1,
\end{equation}
where \(\Lambda_N=0\) for \(N>2r-1\).
\end{cor}

\begin{proof}
For \(r\le m\le 2r-2\), we have
\[
  \frac{\partial\Lambda_m}{\partial c_a}
  =
  -2c_{m-a},
  \qquad
  1\le a\le r-1.
\]
With respect to the ordered variables
\[
  (c_1,\ldots,c_{r-1},\Lambda_{2r-1})
  \quad\text{and}\quad
  (\Lambda_r,\ldots,\Lambda_{2r-2},\Lambda_{2r-1}),
\]
the Jacobian matrix is upper triangular with diagonal entries
\[
  -2c_{r-1},\ldots,-2c_{r-1},1.
\]
Hence, its determinant is
\[
  (-2c_{r-1})^{r-1}.
\]
Thus, the stated variables form rational coordinates after localizing at
\(c_{r-1}\).

In these coordinates, the right-hand side of
\eqref{eq:vector_field_Lambda} acts by
\[
  \Lambda_\ell\longmapsto(\ell-n)\Lambda_{\ell+n}.
\]
This is precisely the defining action of \(\mathcal V_n\) on the coordinates
\(\Lambda_\ell\), and the uniqueness statement in
Theorem~\ref{thm:half_vector_fields} gives the formula.
\end{proof}

\begin{re}[Comparison with previous half-integer realizations]
Formula~\eqref{eq:vector_field_Lambda} is the vector-field part of the
half-integer-rank differential realization written in eigenvalue coordinates
in \cite{HNNT}.  The realizations of \cite{IILZ25,Zang25}, written in
half-integer-indexed parameters
\[
  c_{1/2},c_{3/2},\ldots,c_{r-1/2},
\]
are related to these coordinates by
\[
  \Lambda_m
  =
  -\sum_{a=1/2}^{r-1/2}c_a c_{m-a},
  \qquad r\le m\le 2r-1.
\]
After this change of variables, their vector-field parts also take the form
\eqref{eq:vector_field_Lambda}.  The comparison made here concerns the
vector-field part; the zeroth-order scalar terms are gauge dependent and are
treated separately in Section~\ref{sec:scalar_half}.
\end{re}

\subsection{Construction of the Canonical Operator}

With the vector fields $\mathcal{V}_n$ established, we construct the specific differential operator $L_*$ so that we algebraically isolate the partial derivative $\partial/\partial \Lambda_{2r-1}$ using truncated Virasoro vector fields $\mathcal{V}_0, \dots, \mathcal{V}_{r-1}$.

 The vector fields are represented by the $r \times r$ matrix relation:
\begin{equation}
    \begin{pmatrix}
        \mathcal{V}_0 \\
        \mathcal{V}_1 \\
        \mathcal{V}_2 \\
        \vdots \\
        \mathcal{V}_{r-1}
    \end{pmatrix}
    = \mathcal{M}
    \begin{pmatrix}
        \frac{\partial}{\partial c_1} \\
        \frac{\partial}{\partial c_2} \\
        \vdots \\
        \frac{\partial}{\partial c_{r-1}} \\
        \frac{\partial}{\partial \Lambda_{2r-1}}
    \end{pmatrix}.
\end{equation}
By directly reading the coefficients from the definitions, the matrix $\mathcal{M}$ takes the explicit form:
\begin{equation}\label{eq:matrix_M_anti_triangular}
    \mathcal{M} = 
    \begin{pmatrix}
        c_1 & 2c_2 & \dots & (r-1)c_{r-1} & (2r-1)\Lambda_{2r-1} \\
        h_{1,1} & h_{1,2} & \dots & h_{1,r-1} & 0 \\
        h_{2,1} & h_{2,2} & \dots & 0 & 0 \\
        \vdots & \vdots & \reflectbox{$\ddots$} & \vdots & \vdots \\
        h_{r-1,1} & 0 & \dots & 0 & 0 
    \end{pmatrix}.
\end{equation}
Crucially, the functions $h_{n,k}$ vanish for $n+k > r$, which means $\mathcal{M}_{i,j} = 0$ for all $i+j > r+1$. This demonstrates that the entire matrix $\mathcal{M}$ is anti-upper-triangular. The elements on the anti-diagonal are $\mathcal{M}_{1,r} = (2r-1)\Lambda_{2r-1}$ and $\mathcal{M}_{n+1, r-n} = h_{n, r-n}$ for $1 \le n \le r-1$.

The anti-diagonal entries $h_{n, r-n}$ are evaluated from \eqref{eq:linear_sys_r_minus_1} as  
\begin{equation}
    h_{n, r-n} = - \frac{2r-2n-1}{2c_{r-1}} \Lambda_{2r-1}.
\end{equation}
The determinant of the anti-upper-triangular matrix $\mathcal{M}$ is the product of its anti-diagonal entries multiplied by the sign factor $(-1)^{r(r-1)/2}$. This gives:
\begin{equation}\label{eq:det_M}
    \det \mathcal{M} = (-1)^{\frac{r(r-1)}{2}} (2r-1)\Lambda_{2r-1} \prod_{n=1}^{r-1} \left( - \frac{2r-2n-1}{2c_{r-1}} \Lambda_{2r-1} \right) = \kappa \frac{\Lambda_{2r-1}^r}{c_{r-1}^{r-1}},
\end{equation}
where $\kappa \neq 0$ is a numerical constant. We work over the localization $\C[c_1,\ldots,c_{r-1},c_{r-1}^{-1}]$ and restrict further to the open locus $\Lambda_{2r-1}\neq 0$. 

By inverting the matrix relation, we isolate the derivative with respect to the highest parameter, which is the $r$-th component of the vector of partial derivatives:
\begin{equation}\label{eq:partial_Lambda_M_inv}
    \frac{\partial}{\partial \Lambda_{2r-1}} = \sum_{m=0}^{r-1} (\mathcal{M}^{-1})_{r, m+1} \mathcal{V}_m.
\end{equation}
We define the canonical operator $L_*$ as 
\begin{equation}\label{eq:L_star_def_matrix}
    L_* := \Lambda_{2r-1}^r\sum_{m=0}^{r-1}  (\mathcal{M}^{-1})_{r, m+1} L_m. 
\end{equation}

Because the pole of order $r$ is canceled by the prefactor,  we can uniquely expand the operator $L_*$ as a polynomial in the deformation parameter $\Lambda_{2r-1}$:
\begin{equation}\label{eq:L_star_expansion}
    L_* = f_0 + f_1 \Lambda_{2r-1} + f_2 \Lambda_{2r-1}^2 + \dots + f_{r-1} \Lambda_{2r-1}^{r-1},
\end{equation}
where each $f_i$ is a linear combination of the Virasoro generators $L_0, \dots, L_{r-1}$ with coefficients in the Laurent polynomial ring $\mathbb{C}[c_1, \dots, c_{r-1}, c_{r-1}^{-1}]$. 

\subsection{Canonical Irregular Vector of Half-Integer Rank}

The construction can now be formulated in the parameter-space setting of the
present paper.  The irregular vector of half-integer rank \(r-1/2\) is obtained
from the irregular vector of integer rank \(r-1\).  The base parameter space is
enlarged by one additional continuous parameter \(\Lambda_{2r-1}\), which
serves as the distinguished independent variable.

We construct the canonical irregular vector of half-integer rank as a formal
series in positive powers of \(\Lambda_{2r-1}\), with coefficients in the
rank \(r-1\) irregular Verma module.

\begin{thm}\label{thm:int_to_half}
 For $r\in\Z_{\geq 2}$, there exists a unique formal series
\begin{equation}\label{eq:series_int_to_half}
    |\widetilde I^{(r-1/2)}\rangle = \Lambda_{2r-1}^{\nu} \exp\left(\sum_{j=1}^{r-1} \frac{g_j(c_0, \dots, c_{r-1})}{\Lambda_{2r-1}^j}\right) \sum_{k=0}^\infty \Lambda_{2r-1}^k v_k
\end{equation}
satisfying the defining relations:
\begin{align}
    L_* |\widetilde I^{(r-1/2)}\rangle &= \Lambda_{2r-1}^r \frac{\partial}{\partial \Lambda_{2r-1}} |\widetilde I^{(r-1/2)}\rangle, \label{eq:Lstar_cond_int_to_half} \\
    L_n |\widetilde I^{(r-1/2)}\rangle &= \Lambda_n |\widetilde I^{(r-1/2)}\rangle \quad (n = r, r+1, \dots, 2r-1), \label{eq:Ln_cond_int_to_half}
\end{align}
where $v_0 = |I^{(r-1)}\rangle$ is the irregular vector of rank $r-1$ such that 
\begin{equation*}
    L_nv_0=\Lambda'_n v_0 \quad (n=r-1,r,\ldots,2r-2)
\end{equation*}
with 
\begin{equation*}
    \Lambda'_n=\delta_{n,r-1}((n+1)Q-2c_0)c_{r-1}-\sum_{\substack{p+q=n\\1\leq p,q\leq r-1}} c_p c_{q},
\end{equation*}
and $v_k$ are its descendant states. 
\end{thm}

\begin{proof}
In this proof we set
\[
  \widetilde L_n=L_n-\Lambda'_n \quad (r\le n\le 2r-2),
  \qquad
  \widetilde L_{2r-1}=L_{2r-1},
  \qquad
  \widetilde L_n=L_n \quad (n\ge 2r).
\]

We substitute the ansatz \eqref{eq:series_int_to_half} into the differential equation \eqref{eq:Lstar_cond_int_to_half} and \eqref{eq:Ln_cond_int_to_half}. By expanding both sides in powers of $\Lambda_{2r-1}$ and evaluating the action of $L_*$, we obtain the master recurrence relation for the coefficients: 
\begin{align}
    &\sum_{i=0}^{r-2} \left( f_i+(r-1-i)g_{r-1-i} \right) v_{k-i}
    +\left(f_{r-1}-\nu-k+r-1 \right)v_{k-r+1}=0,\label{eq:L_star_half}\\
    &L_n v_k=\Lambda_n v_k \quad (n=r,r+1,\ldots,2r-2),\label{eq:L_n_half}\\
    &L_{2r-1} v_k=v_{k-1},\label{eq:L_2r-1_half}\\
    &L_n v_k=0 \quad (n\geq 2r). \label{eq:Ln-geq-2r-half}
\end{align}
We use the convention \(v_j=0\) for \(j<0\).
 By the relations \eqref{eq:L_2r-1_half}, we can denote $v_k$ by 
\begin{equation*}
  v_k=\sum_{|\lambda|\leq rk}c_\lambda^{(k)}L_{-\lambda}|I^{(r-1)}\rangle. 
\end{equation*}

First, we construct $v_k$ satisfying the relations \eqref{eq:L_n_half}--\eqref{eq:Ln-geq-2r-half} inductively. 
By definition, we have the matrix equation:
\begin{equation}\label{eq:gram_system}
 \left( \left\{ \widetilde{L}_\mu v_k \right\}\right)_{\mu}=
 \left( \left\{\widetilde{L}_\mu L_{-\lambda} v_0 \right\}\right)_{\mu,\lambda} \left(c_\lambda^{(k)}\right)_\lambda,  
\end{equation}
where $1 \leq |\lambda|, |\mu| \leq rk$. The determinant of the associated Gram matrix is non-zero by Lemma \ref{lem:det}. Consequently, the coefficients $c_\lambda^{(k)}$ for $|\lambda| \geq 1$ are uniquely solved as polynomials in \begin{equation*}
    g_1,\ldots,g_{r-1},\nu,c_0,c_1,\ldots,c_{r-1},  c_{r-1}^{-1},c_\emptyset^{(1)},\ldots,c_\emptyset^{(k-1)}, 
\end{equation*}
  and the relations \eqref{eq:L_n_half}--\eqref{eq:Ln-geq-2r-half} hold.   

Next, we show that the unknown parameters $g_i$ ($i=1,\ldots,r-1$), $\nu$, and $c_\emptyset^{(i)}$ ($i\geq 1$) can be uniquely determined as polynomials in $c_0, c_1, \ldots, c_{r-1}, c_{r-1}^{-1}$.

Setting $k=0$ in \eqref{eq:L_star_half}, we obtain
\begin{equation}
    f_0 v_0 = - (r-1) g_{r-1} v_0.
\end{equation}
By Lemma \ref{lem:f0_constant_term}, the constant term $f_0$ is proportional to $L_{r-1}$ and is non-vanishing as long as $c_{r-1} \neq 0$. Therefore, the unknown function $g_{r-1}$ is uniquely determined. 

For $k \ge 1$, since the algebraic structure of the recurrence \eqref{eq:L_star_half} is completely parallel to that of the integer rank, we can apply the same inductive argument using the truncated states $X_k$. Consequently, $g_{r-k-1}$ (for $k < r-1$), $\nu$ (for $k = r-1$), and $c_\emptyset^{(k-r+1)}$ (for $k > r-1$) are successively and uniquely solved as polynomials in $c_0, c_1, \ldots, c_{r-1}, c_{r-1}^{-1}$.

Finally, we show that the recurrence relation $Y_k = 0$ holds, where $Y_k$ is the left-hand side of the relation \eqref{eq:L_star_half} for $v_k$:
\begin{equation}\label{eq:Y_k_def}
    Y_k = \sum_{i=0}^{r-1} f_i v_{k-i} - (k-r+1+\nu)v_{k-r+1} + \sum_{j=1}^{r-1} j g_j v_{k-r+j+1}.
\end{equation}
Since $g_i$, $\nu$, and $c_\emptyset^{(i)}$ are uniquely solved by looking at the constant terms, we need to show that the actions of $\widetilde{L}_n$ ($n \ge r$) on $Y_k$ are equal to zero. We use induction on $k$, assuming that $v_0, \dots, v_{k-1}$ satisfy $Y_j = 0$ for $j < k$.

Applying $\widetilde{L}_n$ ($n \ge r$) to $Y_k$, and using the relations $\widetilde{L}_n v_j = 0$ ($r \le n \le 2r-2$) and $\widetilde{L}_{2r-1} v_j = v_{j-1}$, we obtain:
\begin{align}
    \widetilde{L}_n Y_k &= \sum_{i=0}^{r-1} [L_n, f_i] v_{k-i} \quad (r \le n \le 2r-2), \label{eq:Ln_act_Yk} \\
    \widetilde{L}_{2r-1} Y_k &= \sum_{i=0}^{r-1} [L_{2r-1}, f_i] v_{k-i} - v_{k-r} + Y_{k-1}. \label{eq:L2r1_act_Yk}
\end{align}
In \eqref{eq:L2r1_act_Yk}, the extra term $-v_{k-r}$ arises because applying $\widetilde{L}_{2r-1}$ to $v_{k-r+1}$ yields $v_{k-r}$, which shifts the index $k$ in the coefficient $(k-r+1+\nu)$ by one compared to $Y_{k-1}$. By the induction hypothesis, $Y_{k-1} = 0$. We note that each $f_i$ does not depend on $k$.

To evaluate the commutator sum, we note that $\sum_{i=0}^{r-1} [L_n, f_i] v_{k-i}$ is the coefficient of $\Lambda_{2r-1}^k$ in the state $[L_n, L_*] |\widetilde I^{(r-1/2)}\rangle$. We compute this using the generating operator:
\begin{equation}
    [L_n, L_*] |\widetilde I^{(r-1/2)}\rangle = \Lambda_{2r-1}^r\sum_{m=0}^{r-1}  (\mathcal{M}^{-1})_{r, m+1} (n-m) L_{n+m} |\widetilde I^{(r-1/2)}\rangle = E(n) |\widetilde I^{(r-1/2)}\rangle,
\end{equation} 
where 
\begin{align*}
    E(n) &= \Lambda_{2r-1}^r \sum_{m=0}^{r-1} (\mathcal{M}^{-1})_{r, m+1} (n-m)\Lambda_{n+m} \\
         &= \Lambda_{2r-1}^r \sum_{m=0}^{r-1} (\mathcal{M}^{-1})_{r, m+1} \mathcal{V}_m(\Lambda_n) \\
         &= \Lambda_{2r-1}^r \frac{\partial \Lambda_n}{\partial \Lambda_{2r-1}}.
\end{align*}

For $r \le n \le 2r-2$, $\Lambda_n$ does not depend on $\Lambda_{2r-1}$, so $\frac{\partial \Lambda_n}{\partial \Lambda_{2r-1}} = 0$. Thus, $E(n) = 0$, which proves that \eqref{eq:Ln_act_Yk} is identically zero. 
For $n = 2r-1$, we have  $\frac{\partial \Lambda_{2r-1}}{\partial \Lambda_{2r-1}} = 1$. Thus, $E(2r-1) = \Lambda_{2r-1}^r$. The coefficient of $\Lambda_{2r-1}^k$ in $E(2r-1) |\widetilde I^{(r-1/2)}\rangle$ is exactly $v_{k-r}$. Substituting this into \eqref{eq:L2r1_act_Yk} yields $v_{k-r} - v_{k-r} = 0$. By construction, the scalar component of \(Y_k\) is zero. 
Moreover, the above computation shows that \(\widetilde L_\mu Y_k=0\) for all relevant partitions \(\mu\). 
Therefore, by Lemma \ref{lem:det}, the non-scalar descendant part of \(Y_k\) also vanishes. Hence \(Y_k=0\). This completes the proof. 

\end{proof}

\subsection{Scalar Gauge and the Full Lower Equations in the Half-Integer Case}\label{sec:scalar_half}

We now prove the analogue of the preceding scalar-gauge statement for
the half-integer-rank vector. 
There is one difference from the integer-rank case.  In the integer-rank
case the scalar parts of the lower deformation operators are fixed in
advance, and the Frobenius theorem is applied only to integrate the
scalar obstructions \(a_i\).  In the half-integer-rank case the
truncated vector fields \(\mathcal V_i\) are first constructed without their
scalar parts.  Hence we first solve a Maurer--Cartan problem for the
unknown scalar functions \(\sigma_i\), with the higher scalar modes
\(S_n\) fixed for \(n\ge r\).  After this scalar completion, the proof
is identical to the integer-rank case.

We first complete the vector-field representation by adding scalar
terms.
\begin{lem}[Scalar completion in the \(L_*\)-gauge]
There exist scalar functions
\[
 \sigma_0,\sigma_1,\ldots,\sigma_{r-1}
\]
in the localized formal coefficient ring such that, if we put
\[
 \mathcal L_i:=\sigma_i+\mathcal V_i\qquad(0\le i\le r-1),
\]
and
\[
 \mathcal L_m:=S_m\qquad(m\ge r),
\]
then
\[
 [\mathcal L_m,\mathcal L_n]=(n-m)\mathcal L_{m+n}
\]
for all \(m,n\ge0\), where \(\mathcal L_N=0\) for \(N>2r-1\).  Moreover,
the scalar functions may be chosen so that
\begin{equation}\label{eq:Lstar-gauge-sigma}
 \sum_{i=0}^{r-1}(\mathcal M^{-1})_{r,i+1}\sigma_i=0.   
\end{equation}
We call this choice the \(L_*\)-gauge.
\end{lem}

\begin{proof}
The scalar terms must satisfy the Maurer--Cartan equations
\begin{equation}\label{eq:sigma}
    \mathcal V_i\sigma_j-\mathcal V_j\sigma_i=(j-i)\sigma_{i+j},
\end{equation}
where we set
\[
 \sigma_N=S_N\quad(N\ge r),\qquad \sigma_N=0\quad(N>2r-1).
\]
After moving the unknown low-mode term
\((j-i)\sigma_{i+j}\) to the left, the scalar completion equations
\eqref{eq:sigma} become
\[
 d_\mathcal V\sigma=\Omega,
\]
where \(d_\mathcal V\) is the Chevalley--Eilenberg differential associated with
the moving frame \(\mathcal V_0,\ldots,\mathcal V_{r-1}\), and where
\[
 \Omega_{ij}=0\quad(i+j<r),\qquad
 \Omega_{ij}=(j-i)S_{i+j}\quad(i+j\ge r).
\]
The compatibility \(d_\mathcal V\Omega=0\)\footnote{Here \(d_\mathcal V\) denotes the Chevalley--Eilenberg differential for the
moving frame \(\mathcal V_0,\ldots,\mathcal V_{r-1}\); in particular,
\[
(d_\mathcal V\sigma)_{ij}
=
\mathcal V_i\sigma_j-\mathcal V_j\sigma_i-(j-i)\sigma_{i+j}
\quad (i+j<r),
\]
and
\[
(d_\mathcal V\sigma)_{ij}
=
\mathcal V_i\sigma_j-\mathcal V_j\sigma_i
\quad (i+j\ge r).
\]} is exactly the Jacobi identity together
with
\[
 \mathcal V_i(S_m)=(m-i)S_{m+i}.
\]
Since \(\mathcal V_0,\ldots,\mathcal V_{r-1}\) form a local frame after localization, the
Chevalley--Eilenberg complex defined by the \(\mathcal V_i\)'s is identified with
the ordinary de Rham complex written in this moving frame.  Hence the
formal Poincaré lemma gives a local \(1\)-cochain
\(\sigma=(\sigma_i)\) solving \(d_\mathcal V\sigma=\Omega\).

There is a scalar gauge freedom for the completed differential
realization.  For a scalar formal function \(\phi\), set
\[
 \mathcal L_i^\phi:=e^\phi\mathcal L_i e^{-\phi}.
\]
For \(0\le i\le r-1\), this gives
\[
 \mathcal L_i^\phi=\mathcal V_i+\sigma_i-\mathcal V_i\phi,
\]
whereas for \(n\ge r\) we have
\[
 \mathcal L_n^\phi=e^\phi S_n e^{-\phi}=S_n.
\]
Thus the higher scalar parts \(\sigma_n=S_n\) for \(n\ge r\) remain
fixed. 
Since this is a conjugation of the differential operators
\(\mathcal L_n\), the Maurer--Cartan equations are preserved.

Set
\[
 B:=\sum_{i=0}^{r-1}(\mathcal M^{-1})_{r,i+1}\sigma_i.
\]
Under the above gauge transformation,
\[
 B^\phi
 =
 B-\sum_{i=0}^{r-1}(\mathcal M^{-1})_{r,i+1}\mathcal V_i\phi
 =
 B-\frac{\partial\phi}{\partial\Lambda_{2r-1}}.
\]
Choosing \(\phi\) such that
\[
 \frac{\partial\phi}{\partial\Lambda_{2r-1}}=B,
\]
we obtain
\[
 \sum_{i=0}^{r-1}(\mathcal M^{-1})_{r,i+1}\sigma_i^\phi=0.
\]

\end{proof}

The analogue of Lemma~\ref{lem:commutativity} also holds in the
half-integer setting.  Indeed, the \(\partial_{\Lambda_{2r-1}}\)-part
commutes trivially with the Virasoro action on the completed rank \(r-1\)
module, while the remaining \(c\)-parameter vector-field part is a rational
linear combination of the lower vector fields of the rank \(r-1\) integer
module.  Hence, by the inductive normalization of \(v_0\),
\[
  [L_n,\mathcal V_i]=0
  \qquad(n\in\mathbb Z,\ 0\le i\le r-1)
\]
on the completed rank \(r-1\) irregular Verma module.

Set
\[
  \mathcal L_* :=
  \Lambda_{2r-1}^r
  \sum_{i=0}^{r-1}(\mathcal M^{-1})_{r,i+1}\mathcal L_i .
\]
In the \(L_*\)-gauge this is simply
\[
  \mathcal L_*=\Lambda_{2r-1}^r
  \frac{\partial}{\partial\Lambda_{2r-1}} .
\]
Define
\[
 R_n:=L_n-\mathcal L_n,\quad R_*:=L_*-\mathcal L_*
\]
with the $L_*$-gauge. 
For \(n\ge r\), this is simply
\[
 R_n=L_n-S_n.
\]

By construction and by the preceding commutativity observation,  
for \(m,n\ge0\), the operators \(R_n\) satisfy
\begin{equation}\label{eq:half_R}
 [R_m,R_n]=(m-n)R_{m+n}.    
\end{equation}

Thanks to Theorem \ref{thm:int_to_half} and \eqref{eq:half_R}, the same argument in the proof of Theorem \ref{thm:integer_full_vector} holds. Therefore, we obtain 
\begin{thm}
Let
\[
 W=|\widetilde I^{(r-1/2)}\rangle
\]
be the canonical half-integer-rank vector constructed in Theorem~\ref{thm:int_to_half}.
Then there exists a non-zero scalar formal function
\[
 f=f(c_0,c_1,\ldots,c_{r-1}), 
\]
 such that
\[
 |I^{(r-1/2)}\rangle:=fW
\]
satisfies
\[
 R_i|I^{(r-1/2)}\rangle=0,
 \qquad 0\le i\le r-1.
\]
Equivalently,
\[
 L_i|I^{(r-1/2)}\rangle
 =
 \mathcal L_i|I^{(r-1/2)}\rangle,
 \qquad 0\le i\le r-1.
\]
The function \(f\) is unique up to multiplication by a non-zero function
of \(c_0\).
\end{thm}
\begin{proof}
The proof is the same as that of Theorem~\ref{thm:integer_full_vector}, with
\(\mathcal D_i\) replaced by the completed half-integer differential operators
\(\mathcal L_i\).  The scalar completion lemma gives the Maurer--Cartan relations for the operators \(R_i=L_i-\mathcal L_i\), and the \(L_*\)-gauge
gives
\[
 R_*W=0.
\]
Thus the obstruction \(R_iW\) to the lower equations is again a scalar
multiple of \(W\).  The Maurer--Cartan relation implies the Frobenius
compatibility of these scalar obstructions, and the \(L_*\)-gauge shows that
the resulting scalar potential is independent of \(\Lambda_{2r-1}\).  Hence it
can be absorbed into a scalar factor \(f(c_0,c_1,\ldots,c_{r-1})\).  The
uniqueness statement follows because the completed vector fields form a local
frame on the lower parameter space.
\end{proof}

\begin{cor}
The scalar gauge factor \(f\) in the preceding theorem can be chosen in
the form
\[
 f
 =
 C(c_0)\,
 c_1^{\nu_1}\cdots c_{r-1}^{\nu_{r-1}}
 \exp\bigl(g_0(c_0,c_1,\ldots,c_{r-1})\bigr),
\]
where \(g_0\) belongs to the same formal Laurent coefficient ring in the
lower parameters, and the exponents \(\nu_1,\ldots,\nu_{r-1}\) depend
only on the passive parameters \(c_0\). 
\end{cor}
\begin{proof}
This follows from the same Laurent-torus decomposition of closed formal
one-forms used in the proof of the integer-rank corollary.
\end{proof}

\section{Discussion}

A consequence of the scalar-gauge argument is that the canonical
\(L_*\)-solution is not merely a solution of the higher-mode recursion.  After
multiplication by a scalar factor depending only on the lower parameters, it
satisfies the full lower Virasoro deformation equations.  This separates the
non-scalar recursive construction from the scalar normalization problem and
clarifies the role of the lower modes in the formal deformation theory.

Recent developments also suggest a broader role of the irregular vectors
constructed in this paper in the study of quantum Painlevé tau functions.
In \cite{BST25}, bilinear equations for quantum Painlevé tau functions were
derived.  It was observed that such tau functions can be written as Zak
transforms of certain building blocks, and that these building blocks can be
computed from the bilinear equations.  The same work further relates these
building blocks to irregular conformal blocks defined by expectation values of irregular vertex operators and pairings of
irregular vectors.

From this point of view, the present construction gives a representation
theoretic framework for several building blocks appearing in the large-time
asymptotic expansions of quantum Painlev\'e tau functions.  The known and
expected realizations may be summarized as follows:

\begin{center}
\renewcommand{\arraystretch}{1.15}
\begin{tabular}{c|c|l|c}
Equation & Singularity type & CFT realization & Reference\\
\hline
\(\mathrm{P}_{\mathrm{V}}\) & \((0,0,1)\) &
irregular vertex operator & \cite{Nagoya15}\\
\(\mathrm{P}_{\mathrm{IV}}\) & \((0,2)\) &
irregular vertex operator & \cite{Nagoya15}\\
\(\mathrm{P}_{\mathrm{IV}}\) & \((0,2)\) &
irregular vector & this paper\\
\(\mathrm{P}_{\mathrm{III}_{1}}\) & \((0,0,\frac{1}{2})\) &
ramified irregular vertex operator & \cite{Nagoya18}\\
\(\mathrm{P}_{\mathrm{II}}\) & \((0,\frac{3}{2})\) &
ramified irregular vertex operator & \cite{Nagoya18}\\
\(\mathrm{P}_{\mathrm{II}}\) & \((0,\frac{3}{2})\) &
irregular vector & this paper\\
\(\mathrm{P}_{\mathrm{II}}\) & \((3)\) &
irregular vector & this paper\\
\(\mathrm{P}_{\mathrm{I}}\) & \((\frac{5}{2})\) &
irregular vector & this paper
\end{tabular}
\end{center}

All entries in the table correspond to asymptotic expansions associated
with irregular singularities at infinity.

We expect that the present construction will serve as a representation-theoretic
basis for the asymptotic and resurgent analysis of irregular conformal blocks,
Argyres--Douglas partition functions, and related accessory parameters for
Heun equations of confluent type \cite{INS26}.  It is also natural to expect
that the remaining building blocks in the quantum Painlev\'e theory admit
similar constructions from Virasoro representation theory.

\appendix

\section{Lemmas}

To determine the entries of the inverse matrix $M^{-1}$, we introduce the following lemma.

\begin{lem} \label{lem:inverse_A}
Let $A(z)$ be a formal power series defined by
\begin{equation}
    A(z) = \frac{1}{c_r + c_{r-1}z + \cdots + c_1z^{r-1}} = \sum_{j=0}^\infty a_j z^j.
\end{equation}
Then, the entries of the $r$-th row of the inverse matrix $M^{-1}$ are given by
\begin{equation}
    (M^{-1})_{r, k} = \frac{1}{r} a_{k-1}
\end{equation}
for each $1 \le k \le r$.
\end{lem}

\begin{proof}
Let $\mathbf{y} = (y_1, \ldots, y_r)$ denote the $r$-th row vector of $M^{-1}$. By the definition of the inverse matrix, we have $\sum_{i=1}^r y_i M_{i, k} = \delta_{r, k}$ for any $1 \le k \le r$. Substituting $M_{i,k} = k c_{i+k-1}$ (where $i+k-1 \le r$), the sum ranges from $i=1$ to $r-k+1$, yielding 

\begin{equation*}
    \sum_{i=1}^{r-k+1} y_i k c_{i+k-1} = \delta_{r, k}.  
\end{equation*}
Since $k=r$ whenever $\delta_{r,k}\neq 0$, this system is equivalent to 
\begin{equation*}
    \sum_{i=1}^{r-k+1} r y_i c_{i+k-1} = \delta_{r, k}.
\end{equation*}

Letting $p = r-k$, the index $p$ runs from $r-1$ down to $0$ as $k$ goes from $1$ to $r$. Rewriting the above equation in terms of $p$, we obtain
\begin{equation}
    \sum_{i=1}^{p+1} r y_i c_{r-p+i-1} = \delta_{p, 0}.
\end{equation}
On the other hand, consider the identity $A(z) \sum_{m=0}^{r-1} c_{r-m} z^m = 1$, which is a rearrangement of the defining equation of $A(z)$. Comparing the coefficients of $z^p$, we find $\sum_{m=0}^p a_m c_{r-(p-m)} = \delta_{p, 0}$. Shifting the summation index by setting $i = m+1$, this becomes $\sum_{i=1}^{p+1} a_{i-1} c_{r-p+i-1} = \delta_{p, 0}$, which exactly coincides with the equation that $r y_i$ must satisfy. Therefore, we conclude that $r y_i = a_{i-1}$.
\end{proof}

Let us recall that we set $L_*$ for integer rank as   
\begin{equation*}
    L_*=\sum_{i=0}^{r-1}f_ic_r^i,  
\end{equation*}
where each \(f_i\) is a linear combination of \(D_0,\ldots,D_{r-1}\)
with coefficients in \(\mathbb C[c_1,\ldots,c_{r-1}]\).

\begin{lem}\label{lem:f0}
    We have 
    \begin{equation}
        f_0 = \frac{(-1)^{r-1}}{r} c_{r-1}^{r-1} D_{r-1}. 
    \end{equation}
\end{lem}
\begin{proof}
From Lemma \ref{lem:inverse_A}, we have
\begin{equation}
    L_* = \frac{c_r^r}{r} \sum_{n=0}^{r-1} a_n D_n.
\end{equation}
To determine the constant term $f_0$ (i.e., the coefficient of $c_r^0$ in $L_*$), we examine the dependence of the expansion coefficients $a_n$ on $c_r$.

Factoring out $1/c_r$ and expanding $A(z)$ as a geometric series, we obtain
\begin{equation}
    A(z) = \frac{1}{c_r} \left( 1 + \sum_{m=1}^{r-1} \frac{c_{r-m}}{c_r} z^m \right)^{-1} = \frac{1}{c_r} \sum_{k=0}^\infty (-1)^k \left( \sum_{m=1}^{r-1} \frac{c_{r-m}}{c_r} z^m \right)^k.
\end{equation}
When determining the coefficient $a_n$ of $z^n$, we note that the lowest degree term in the parenthesis is $z^1$ (with coefficient $c_{r-1}/c_r$). Therefore, to generate $z^n$, the power $k$ in the sum can be at most $k=n$. The term with this maximum power $k=n$ arises from the unique combination of choosing the $z^1$ term $n$ times, yielding a contribution of
\begin{equation}
    \frac{1}{c_r} (-1)^n \left( \frac{c_{r-1}}{c_r} \right)^n = (-1)^n \frac{c_{r-1}^n}{c_r^{n+1}}.
\end{equation}
Any contribution arising from terms with $k < n$ will have a strictly lower degree of $c_r$ in the denominator than $n+1$. Thus, the structure of $a_n$ is determined as
\begin{equation}
    a_n = (-1)^n \frac{c_{r-1}^n}{c_r^{n+1}} + \mathcal{O}(c_r^{-n}).
\end{equation}
Multiplying this by $c_r^r$, which is the overall prefactor of $L_*$, yields
\begin{equation}
    c_r^r a_n = (-1)^n c_{r-1}^n c_r^{r-n-1} + \mathcal{O}(c_r^{r-n}).
\end{equation}
The constant term $f_0$ corresponds to the $c_r^0$ term of this expression. Since the range of the sum is $0 \le n \le r-1$, the exponent $r-n-1$ is always non-negative. It is clear that the $c_r^0$ (constant) term appears strictly when $r-n-1 = 0$, that is, when $n = r-1$. For all $n < r-1$, the terms are of order $\mathcal{O}(c_r)$ and do not contribute to $f_0$ at all.

From the above, by extracting only the $n=r-1$ term, we obtain
\begin{equation}
    f_0 = \frac{(-1)^{r-1}}{r} c_{r-1}^{r-1} D_{r-1}.
\end{equation}
\end{proof}

\begin{lem}\label{lem:f0_constant_term}
Let
\[
L_*=\sum_{i=0}^{r-1} f_i\Lambda_{2r-1}^i
\]
be the polynomial expansion of \(L_*\) in \(\Lambda_{2r-1}\).
Then the constant term \(f_0\) is a non-zero scalar multiple of
\(c_{r-1}^{2r-2}L_{r-1}\). In particular, \(f_0\) is proportional to
\(L_{r-1}\) and is non-zero on the open locus \(c_{r-1}\ne0\).
\end{lem}
\begin{proof}
By Cramer's rule,
\[
(\mathcal M^{-1})_{r,m+1}
=
\frac{(-1)^{m+1+r}\det \mathcal M(m+1|r)}
{\det \mathcal M}.
\]
After multiplication by \(\Lambda_{2r-1}^r\), the constant term can only
come from the leading \(\Lambda_{2r-1}^r\)-part of \(\det\mathcal M\)
and the constant term of the corresponding cofactor.

For \(m<r-1\), the cofactor
\(\det\mathcal M(m+1|r)\) vanishes at \(\Lambda_{2r-1}=0\), because the
resulting minor contains a zero row/column in the anti-triangular limit.
Thus no \(L_0,\ldots,L_{r-2}\) term contributes to \(f_0\).

For \(m=r-1\), the cofactor \(\det\mathcal M(r|r)\) at
\(\Lambda_{2r-1}=0\) is the coefficient matrix of the integer-rank
\(r-1\) vector fields, with highest parameter \(c_{r-1}\). Its
determinant is non-zero and proportional to \(c_{r-1}^{r-1}\). Since
\(\det\mathcal M\) has leading term proportional to
\(\Lambda_{2r-1}^r/c_{r-1}^{r-1}\), the constant term of
\(\Lambda_{2r-1}^r(\mathcal M^{-1})_{r,r}\) is a non-zero scalar
multiple of \(c_{r-1}^{2r-2}\). Hence \(f_0\) is a non-zero scalar
multiple of \(c_{r-1}^{2r-2}L_{r-1}\).
\end{proof}

\section*{Statements and Declarations}

\paragraph{Funding}
This work was supported by JSPS KAKENHI Grant Number 22K03350.

\paragraph{Competing interests}
The author declares no competing interests.

\paragraph{Data availability}
No datasets were generated or analysed during the current study.

\end{document}